\documentclass[draftclsnofoot, onecolumn, comsoc, 12pt]{IEEEtran} 
\IEEEoverridecommandlockouts

\usepackage{comment}
\usepackage{graphicx}
\usepackage[group-separator={,},group-minimum-digits=4]{siunitx}
\usepackage{lipsum}
\usepackage{algorithm}
\usepackage{color, colortbl}
\definecolor{Gray}{gray}{0.9}
\definecolor{White}{gray}{1}
\usepackage{makecell} 
\usepackage{array}
\newcolumntype{P}[1]{>{\centering\arraybackslash}m{#1}}
\newcolumntype{A}{>{\columncolor{White}}P}
\usepackage{algpseudocode}

\usepackage{cite}
\usepackage{subfigure}
\usepackage{multirow}
\usepackage{amsmath,amssymb,amsfonts}
\usepackage{float}
\usepackage{mathtools}

\usepackage{amsthm}
\usepackage{graphicx}
\usepackage{textcomp}
\usepackage[utf8]{inputenc}
\usepackage[english]{babel}
\usepackage{booktabs}
\usepackage[noabbrev]{cleveref}
\usepackage[short,nocomma]{optidef}
\def\BibTeX{{\rm B\kern-.05em{\sc i\kern-.025em b}\kern-.08em
    T\kern-.1667em\lower.7ex\hbox{E}\kern-.125emX}}
\newtheorem{prop}{Proposition}
\newtheorem{remark}{Remark}
\newcommand*{\tcr}[1]{\textcolor{black}{#1}}
\def\tcb{\textcolor{black}}
\def\tcr{\textcolor{black}}
\def\tcbb{\textcolor{black}}

\begin{document}
\bstctlcite{IEEEexample:BSTcontrol}

\title{Beyond Beam Sweeping: One-Shot Satellite Acquisition with Doppler-Aware Rainbow Beamforming}


\author{Juha Park, Ian P. Roberts, and Wonjae Shin
    \thanks{
    }
    \thanks{J. Park and W. Shin are with the School of Electrical Engineering, Korea University, Seoul 02841, South Korea 
    (email: {\texttt{\{juha, wjshin\}@korea.ac.kr}}).}
    \thanks{I. P. Roberts is with the Department of Electrical and Computer Engineering, UCLA, Los Angeles, CA 90095, USA (e-mail: ianroberts@ucla.edu).}
    }
\maketitle
\begin{abstract}
High-gain beamforming (BF) is essential for low Earth orbit (LEO) satellite communications to overcome severe path loss, but this requires acquiring precise satellite positions. Conventional satellite acquisition typically relies on time-domain beam sweeping, which incurs substantial overhead and latency. \tcb{In this correspondence, we propose an efficient one-shot satellite acquisition framework that capitalizes on two phenomena traditionally regarded as impairments: i) Doppler effects and ii) beam-squint effects.} Specifically, we derive a closed-form \emph{rainbow beamformer} that leverages beam-squint effects to align frequency-dependent beam directions with satellite positions inferred from their Doppler shifts. 
 This approach enables reception from multiple satellites {at once} without requiring beam sweeping. To extract satellite position information, we develop three Doppler-aware angle estimation algorithms based on received signals. Simulation results demonstrate that the proposed method significantly outperforms conventional beam sweeping approaches in both acquisition accuracy and required time slots. These gains stem from the ability of the proposed rainbow BF to exploit the \emph{angle-dependent nature of Doppler shifts}, enabling full angular-domain coverage with a single pilot transmission and reception.
\end{abstract}

\section{Introduction}
Low Earth orbit (LEO) satellite communication (SATCOM) has emerged as a key enabler of global coverage and high data rates in 5G and beyond \cite{kodheli2020satellite}. To overcome the severe path loss resulting from the long propagation distances, LEO SATCOM systems require precise high-gain beamforming (BF). Achieving such high BF gain necessitates the use of large antenna arrays (e.g., Starlink's \num{1280}-element phased array \cite{yang2023starlink}), which form narrow pencil beams. As a result, acquiring an estimate on the position of the satellite within the field of view (FOV) becomes critical. This satellite acquisition process is essential not only during the initial access phase to discover available satellites for communication but also during the data communication phase to support seamless handover. However, conventional analog BF architectures are typically constrained by a limited number of available radio frequency (RF) chains, which restricts the number of simultaneously steerable beams. 
 As a result, time-domain beam sweeping using \tcr{a small set of simultaneous beams has traditionally been} employed for satellite acquisition. Despite its practicality, this method incurs significant system overhead and latency due to its sequential search nature. 
\tcb{To address this limitation, we turn two phenomena, typically regarded as impairments in LEO SATCOM systems, into enablers: i) Doppler effects and ii) beam-squint effects.} Rather than compensating for these effects, our approach embraces them to facilitate efficient one-shot satellite acquisition. 

The high velocity of LEO satellites induces substantial Doppler shifts, up to approximately \num{450} kHz at a \num{20} GHz carrier frequency and 
 \num{500} km orbital altitude. These Doppler shifts in LEO SATCOM can cause severe inter-carrier interference (ICI) in orthogonal frequency division multiplexing (OFDM) systems \cite{1638663}, where frequency offsets lead to overlap between adjacent subcarriers. 
 {Existing research on LEO SATCOM} has primarily focused on mitigating the Doppler effect for reliable communications. For instance, \cite{10464931} proposed a Doppler compensation technique in which satellites pre-compensate the Doppler shift at the beam center, while the user terminal post-compensates the residual Doppler shift.

Analog BF architectures based on phase shifters (PS) suffer from beam-squint effects, wherein distinct frequency components inadvertently steer to slightly different directions \cite{dai2022delay,you2022beam}. This frequency-dependent deviation can lead to {misalignment in beam} {pointing}, particularly in wideband systems. {The} {majority of studies on analog/hybrid BF have focused on mitigating the beam-squint effects by utilizing advanced baseband processing \cite{mewes2025beam}.} \tcb{Recent works \cite{ratnam2022joint,10914514} have begun exploring how this effect can be constructively leveraged in terrestrial systems through true time delay (TTD) architectures, enabling frequency-dependent BF, also known as \emph{rainbow BF}.}

\tcb{In addition to Doppler effects and beam-squint effects, majority of prior research has focused on analyzing the impact of various hardware impairments on system performance and mitigating them, such as RF chain impairments \cite{li2024stacked} and mutual coupling \cite{li2025holographic}.} \tcb{In contrast to prior works that treat hardware impairments as performance-degrading factors, this correspondence proposes a novel satellite acquisition framework that constructively exploits these impairments to dramatically reduce acquisition overhead in LEO satellite communications. First, we derive a closed-form rainbow beamformer that aligns frequency-dependent beam directions with satellite positions based on their corresponding Doppler frequency, by controlling the beam-squint effects. This approach enables simultaneous reception and amplification of pilot signals from multiple satellites, each experiencing a distinct Doppler shift. Subsequently, we introduce three Doppler-aware angle estimation algorithms to extract satellite directions from the received signals by exploiting the angle-dependent nature of Doppler shifts. By doing so, the angular positions of multiple satellites can be acquired with a single pilot transmission and reception, referred to as one-shot estimation throughout this correspondence.}
\subsection{Notation}
Scalars, vectors, and matrices are denoted by standard letters, lower-case boldface letters, and upper-case boldface letters, respectively. The imaginary unit is $j\triangleq\sqrt{-1}$. The notations $(\cdot)^{\sf{T}}$ and $(\cdot)^{\sf{H}}$ denote the transpose and conjugate transpose, respectively. $\mathbb{E}[\cdot]$ denotes the statistical expectation. $\mathcal{CN}(\mu,\sigma^2)$ represents a complex Gaussian random variable with mean $\mu$ and variance $\sigma^2$. $\Re\{\cdot\}$ denotes the real part of a complex number and is applied element-wise for vectors. $\mathbf{x}^{(n)}$ represent the $n$-th element of vector $\mathbf{x}$. $\lfloor x \rfloor$ denotes the greatest integer less than or equal to a given real number $x$.

\section{System Model}
\begin{figure}[t]
\centering
\includegraphics[width=0.8\linewidth]{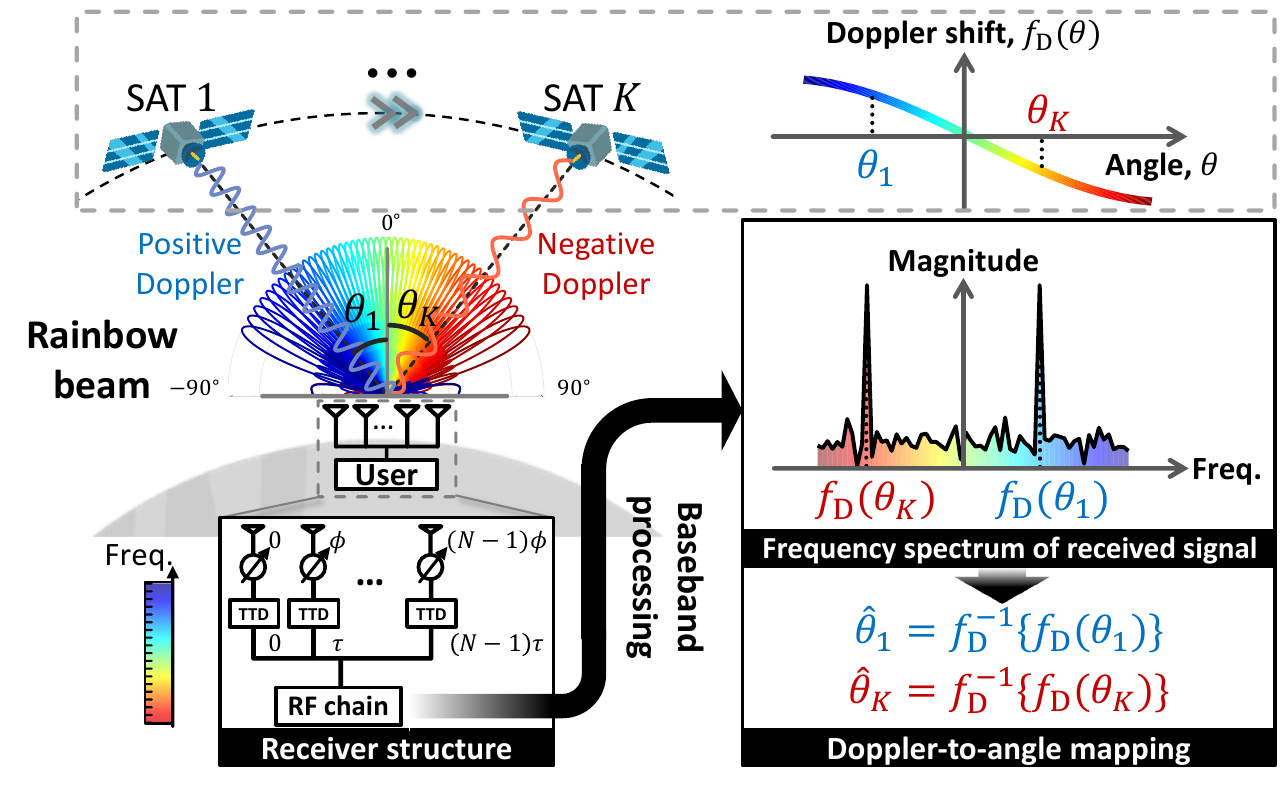}
\caption{System model and schematic of the proposed rainbow BF-empowered satellite acquisition {framework.}}
\label{system model}
\end{figure}

We consider a joint phase-time array (JPTA) architecture \cite{ratnam2022joint} with a uniform linear array (ULA) at the user terminal\footnote{\tcr{Although a ULA is considered for simplicity, the proposed framework can be extended to other array configurations by generalizing the array geometry and the corresponding array steering vector.}}, consisting of $N$ half-wavelength spaced antenna elements operating at center frequency $f_{\sf c}$ \cite{ratnam2022joint}. Each antenna element is individually connected to a dedicated TTD and phase shifter (PS), as illustrated in Fig. \ref{system model}. {Although both parameters may be generally reconfigurable,} the phase shift and delay values for the $n$-th antenna element are \tcr{fixed} to $(n-1)\phi$ and $(n-1)\tau$, respectively, i.e., a uniformly weighted array with progressive delay $\tau$ and progressive phase shift $\phi$. 
{The boresight direction of the user terminal is defined as $0^\circ$ and is assumed to coincide with the zenith for analytical convenience.} We denote the minimum elevation angle (i.e., the lowest elevation angle at which a satellite can be detected) as $\beta$. Accordingly, the FOV is defined as the angular range $\boldsymbol{\Psi} \triangleq [-\frac{\pi}{2}+\beta,\frac{\pi}{2}-\beta]$. Suppose there are $K$ satellites simultaneously transmitting single-tone pilot signals at frequency $f_{\sf p}$. The off-boresight angle of the $k$-th satellite relative to the user terminal is denoted by $\theta_k\in\boldsymbol{\Psi}$. For a specific pilot transmission and reception duration (i.e., a time slot), the passband signal transmitted by 
the $k$-th satellite is modeled as $\Re\{x_ke^{j2\pi f_{\sf p} t}\}$, where $x_k$ is the pilot symbol and $t$ is the continuous time index.

Let $\mathbf{h}_k\in\mathbb{C}^{N\times1}$ denote the flat-fading baseband channel vector from the $k$-th satellite to the user. The received passband signal at the user terminal, $\mathbf{y}_{\sf PB}(t)\in\mathbb{R}^{N\times1},$ is expressed as
\begin{equation}
\mathbf{y}_{\sf PB}(t) = \Re\left\{\sum_{k=1}^K\mathbf{h}_k x_ke^{j2\pi\left\{{f_{\sf p}}+f_{\sf D}(\theta_k)\right\}t}\right\}+\mathbf{z}_{\sf PB}(t),
\end{equation}
where $f_{\sf D}\left(\theta\right)$ is the angle-dependent Doppler frequency due to satellite motion, and $\mathbf{z}_{\sf PB}(t)$ denotes Gaussian noise. \tcb{We adopt a widely used LEO SATCOM channel model \cite{you2020massive,you2022beam,10680149}, in which $\mathbf{h}_k=g_k(f_{\sf p}+f_{\sf D}(\theta_k))\mathbf{a}\left(f_{\sf p}+f_{\sf D}\left(\theta_k\right),\theta_k\right)$, where $g_k(f_{\sf p}+f_{\sf D}(\theta_k))$ is the complex channel gain and  $\mathbf{a}\left(f,\theta\right)\triangleq[1,e^{-j\pi\frac{f}{f_c}\sin\theta}, \cdots,e^{-j\pi\frac{f}{f_c}(N-1)\sin\theta}]^{\sf T}$ is the frequency-dependent array response vector.}\footnote{The array response model captures the frequency-dependent nature of the array response that leads to the beam-squint effect. Under a narrowband assumption, i.e., $\frac{f}{f_{\sf c}} \approx 1$, the array response vector reduces to the conventional array response,
$\mathbf{a}\left(f,\theta\right)\approx[1,e^{-j\pi\sin\theta}, \cdots,e^{-j\pi(N-1)\sin\theta}]^{\sf T}$, \tcr{$\forall f$, which is no longer a function of frequency.}} \tcb{After passing through the PS and TTD units, the received signal at the $n$-th antenna element (neglecting noise) can be expressed as
\begin{align}
&\Re\left\{\sum_{k=1}^K\mathbf{h}_k^{(n)} x_k e^{j2\pi\left\{f_{\sf p}+f_{\sf D}(\theta_k)\right\}\left\{t-(n-1)\tau\right\}}e^{j(n-1)\phi}\right\}
\nonumber\\
=&\Re\left\{\sum_{k=1}^K\mathbf{h}_k^{(n)} x_k \underbrace{e^{j(n-1)\left[\phi-2\pi\left\{f_{\sf p}+f_{\sf D}(\theta_k)\right\}\tau\right]}}_{\text{JPTA response for}\,n\text{-th element}} e^{j2\pi\left\{f_{\sf p}+f_{\sf D}(\theta_k)\right\}t}\right\},
\end{align} where $\mathbf{h}_k^{(n)}$ denotes $n$-th element of $\mathbf{h}_k$.}
\tcb{Here, we observe that the expression exhibits the \emph{frequency-dependent phase shifts} of the JPTA.} 
\tcb{The discrete-time received complex baseband signal after signal combining (receive BF), down-converting to $f_{\sf p}$, and sampling is given by}
\begin{align}
y[\ell]=\sum_{k=1}^K \tilde{g}_k x_k e^{j2\pi\frac{f_{\sf D} \left(\theta_k\right)}{{f_{\sf s}}}\ell}+z[\ell], \ell\in\{0,\cdots,L-1\},
\label{Beamformed RX signal}
\end{align}
\tcb{where $\tilde{g}_k\triangleq g_k(f_{\sf p}+f_{\sf D}(\theta_k))\mathbf{w}(f_{\sf p}+f_{\sf D} (\theta_k),\tau,\phi)^{\sf H}\mathbf{a}\left(f_{\sf p}+f_{\sf D} (\theta_k),\theta_k\right)$ is the effective channel coefficient with arguments omitted for brevity;} $\ell$ is the discrete time index; $f_{\sf s}$ denotes the sampling frequency; $\mathbf{w}(f,\tau,\phi)\triangleq[1,e^{j\{\phi-2\pi f \tau\}},\cdots,e^{j(N-1)\{\phi-2\pi f \tau\}}]^{\sf H}$ denotes the frequency-dependent JPTA beamformer; $z[\ell]\sim\mathcal{CN}(0,\sigma^2)$ is the effective complex Gaussian noise.

\subsection{Angle-Dependent Doppler Shift Model}
\begin{figure}[t]
\centering
\includegraphics[width=0.7\linewidth]{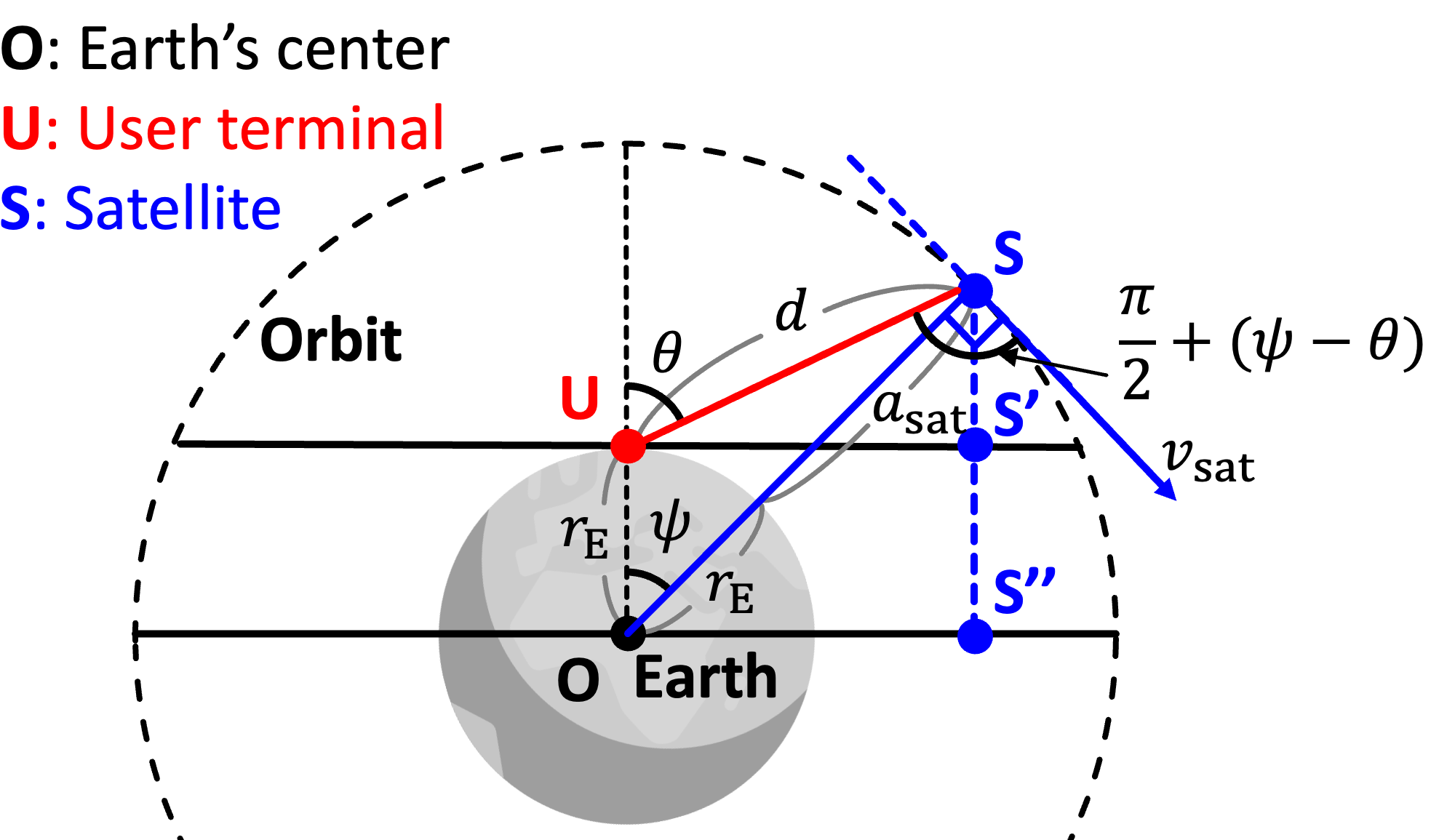}
\caption{Satellite-user-Earth geometry for Doppler model derivation.}
\label{Doppler_geometry}
\end{figure}

We \tcr{now} derive the angle-dependent Doppler shift model to characterize $f_{\sf D}({\theta})$. Based on the satellite-user-Earth geometry in Fig. \ref{Doppler_geometry}, we begin with the standard Doppler shift formula \cite{10694172}:
\begin{equation}
f_{\sf D}(\theta) = \frac{f_{\sf p}v_{\sf sat}}{c}\cos\left(\frac{\pi}{2}+(\psi-\theta)\right) = -\frac{f_{\sf p}v_{\sf sat}}{c}\sin\left(\psi-\theta\right),
\label{Basic Doppler Formula}
\end{equation}
where $c$ denotes the speed of light; $v_{\sf sat}$ is the satellite velocity; $\psi$ represents the satellite angle with respect to Earth's center. 
Applying the law of sines to triangle $\mathrm{SUO}$ (see Fig. \ref{Doppler_geometry}) yields
\begin{equation}
\frac{\sin(\psi-\theta)}{r_{\sf E}}=\frac{\sin\psi}{d},
\label{Sine law}
\end{equation}
where $r_{\sf E}$ and $d$ are Earth's radius and the satellite-to-user distance, respectively.
Utilizing $\overline{\rm US'}=\overline{\rm OS''}$, we have
\begin{equation}
d\sin\theta=(r_{\sf E}+a_{\sf sat})\sin\psi\,\,\,\Leftrightarrow	\,\,\,
\frac{\sin\psi}{d}=\frac{\sin\theta}{r_{\sf E}+a_{\sf sat}},
\label{Sine law 2}
\end{equation}
where $a_{\sf sat}$ is \tcr{the} orbital altitude. Finally, substituting Eq. \eqref{Sine law} and Eq. \eqref{Sine law 2} into Eq. \eqref{Basic Doppler Formula}, the Doppler shift simplifies to a sinusoidal function of the satellite angle $\theta$: 
\begin{equation}
f_{\sf D}(\theta) = -\alpha\sin\theta,\label{Doppler equation}
\end{equation}
where $\alpha\triangleq\frac{f_{\sf p}v_{\sf sat}r_{\sf E}}{c(r_{\sf E}+a_{\sf sat})}$.

\section{A Closed-Form Rainbow Beamformer Design}
We design a beamformer that achieves near-maximum BF gain for pilot signals received with distinct Doppler shifts, regardless of the signals' angle-of-arrival (AOA). The frequency-dependent BF gain is defined as 
\begin{align}
G(f,\theta,\tau,\phi)\triangleq\frac{1}{N}\left|\mathbf{w}(f,\tau,\phi)^{\sf H}\mathbf{a}\left(f,\theta\right)\right|^2.
\end{align} Accordingly, the BF gain for the signal received from a satellite located at angle $\theta$ is given by
\begin{align}
&G\left(f_{\sf p}+f_{\sf D}(\theta),\theta,\tau,\phi\right)
\nonumber\\
&=\frac{1}{N}\left|\sum_{n=1}^N e^{j(n-1)\left[\phi-2\pi\left\{f_{\sf p}+f_{\sf D}\left(\theta\right)\right\}\tau-\pi\frac{f_{\sf p}+f_{\sf D}\left(\theta\right)}{f_{\sf c}}\sin\theta\right]}\right|^2
\nonumber\\
&=\frac{\sin^2 \left[\frac{N}{2}\left\{\phi-2\pi\left\{f_{\sf p}+f_{\sf D}\left(\theta\right)\right\}\tau-\pi\frac{f_{\sf p}+f_{\sf D}\left(\theta\right)}{f_{\sf c}}\sin\theta\right\}\right]}{N\sin^2 \left[\frac{1}{2}\left\{\phi-2\pi\left\{f_{\sf p}+f_{\sf D}\left(\theta\right)\right\}\tau-\pi\frac{f_{\sf p}+f_{\sf D}\left(\theta\right)}{f_{\sf c}}\sin\theta\right\}\right]}
\nonumber\\
&\overset{(\rm{a})}{\leq} N,
\label{beamforming gain}
\end{align}
where {$\mathrm{(a)}$} holds with equality, i.e., the BF gain reaches its maximum at angle $\theta$, if and only if
\begin{align}
\phi-2\pi\left(f_{\sf p}+f_{\sf D}\left(\theta\right)\right)\tau-\pi\frac{f_{\sf p}+f_{\sf D}\left(\theta\right)}{f_{\sf c}}\sin\theta=2\gamma\pi,
\label{max_gain_condition}
\end{align}
with arbitrary integer $\gamma\in\mathbb{Z}$. 

Our beamformer design can be mathematically formulated as the following feasibility problem:
\begin{align}
&{\rm find}\quad\tau,\,\phi
\nonumber\\
&{\rm s.t.}\quad G\left(f_{\sf p}+f_{\sf D}(\theta),\theta,\tau,\phi\right)=N,\,\forall \theta\in\boldsymbol{\Psi}.
\label{Feasiblity problem}
\end{align} This formulation imposes infinitely many equality constraints over the angular domain $\boldsymbol{\Psi}$. \tcr{Solving problem \eqref{Feasiblity problem} is equivalent to finding $\{\tau,\phi\}$ that satisfy the maximum gain condition in Eq. \eqref{max_gain_condition} for all $\theta\in\boldsymbol{\Psi}$, while keeping $\{\tau,\phi\}$ constant.} By substituting the Doppler model from Eq. \eqref{Doppler equation} into Eq. \eqref{max_gain_condition}, we obtain
\begin{align}
\phi-2\pi\left(f_{\sf p}-\alpha\sin\theta\right)\tau-\pi\frac{f_{\sf p}-\alpha\sin\theta}{f_{\sf c}}\sin\theta=2\gamma\pi, \forall\theta\in\boldsymbol{\Psi}.
\label{equivalent condition for maximum gain}
\end{align} 
This equation involves both $\sin\theta$ and $\sin^2\theta$, and thus satisfying it for all $\theta$ requires that the coefficients of both terms be zero. However, the presence of the non-zero coefficient $\pi\frac{\alpha}{f_{\sf c}}$ for  $\sin^2\theta$ makes an exact solution for all $\theta$ infeasible. Since the Doppler frequency $-\alpha\sin\theta$ is on the order of tens to hundreds of kHz, significantly smaller than the pilot frequency $f_{\sf p}$ in the GHz range, we apply the approximation $\frac{f_{\sf p}-\alpha\sin\theta}{f_{\sf c}}\approx\frac{f_{\sf p}}{f_{\sf c}}$, yielding a simplified expression:
\begin{align}
\left(2\pi\alpha\tau-\pi\frac{f_{\sf p}}{f_{\sf c}}\right)\sin\theta+\phi-2\pi f_{\sf p}\tau=2\gamma\pi, \nonumber\\
\forall \theta\in\boldsymbol{\Psi}, \gamma\in\mathbb{Z}.
\label{Approximated equation}
\end{align} 
As a result, the equation reduces to a linear function in $\sin\theta$. The optimal progressive delay $\tau^\star$ that nullifies the coefficient of $\sin\theta$ is thus given by 
\begin{align}
    \tau^\star=\frac{f_{\sf p}}{2f_{\sf c}\alpha}=\frac{c(r_{\sf E}+a_{\sf sat})}{2f_{\sf c}v_{\sf sat}r_{\sf E}}.
\label{optimal delay}
\end{align}
Substituting this expression for $\tau^{\star}$ into Eq. \eqref{Approximated equation} yields the optimal progressive phase shift $\phi^{\star}$ as follows:
\begin{align}
\phi^\star=\frac{\pi f_{\sf p}^2}{\alpha f_{\sf c}}+2\gamma\pi=\frac{\pi f_{\sf p}c(r_{\sf E}+a_{\sf sat})}{f_{\sf c}v_{\sf sat}r_{\sf E}}+2\gamma\pi, \gamma\in\mathbb{Z}.
\label{optimal phase shift}
\end{align}
Due to the periodicity of phase, $\phi^{\star}$ can be expressed without loss of generality as
$\phi^{\star}=\big(\frac{\pi f_{\sf p}c(r_{\sf E}+a_{\sf sat})}{f_{\sf c}v_{\sf sat}r_{\sf E}}{\rm ~mod~}2\pi\big)$.

\begin{prop} 
The BF gain loss, i.e., the deviation from the maximum BF gain when using the beamformer from Eqs. \eqref{optimal delay}–\eqref{optimal phase shift} at angle $\theta$, can be approximated by
\begin{align}
\delta(\theta)&\triangleq \left\vert N-G(f_{\sf p}+f_{\sf D}(\theta),\theta,\tau^\star,\phi^\star)\right\vert
\nonumber\\
&\approx \frac{\pi^2f_{\sf p}^2v_{\sf sat}^2r_{\sf E}^2}{12f_{\sf c}^2c^2(r_{\sf E}+a_{\sf sat})^2}N^3\sin^4\theta,
\label{Beam gain loss}
\end{align}
where the worst-case beam gain loss occurs at $\theta=\pm \big(\frac{\pi}{2}-\beta\big)$. This validates the approximation $\frac{f_{\sf p}-\alpha\sin\theta}{f_{\sf c}}\approx\frac{f_{\sf p}}{f_{\sf c}}$ used in the beamformer design for practical LEO SATCOM scenarios. 
\end{prop}

\begin{proof} 
Let us define $x\triangleq \phi-2\pi\left\{f_{\sf p}+f_{\sf D}\left(\theta\right)\right\}\tau-\pi\frac{f_{\sf p}+f_{\sf D}\left(\theta\right)}{f_{\sf c}}\sin\theta$ for notational brevity. Then, we can find the approximated beam gain loss as follows:
\begin{align}
\delta(\theta) &= \left\vert N-G(f_{\sf p}+f_{\sf D}(\theta),\theta,\tau^\star,\phi^\star)\right\vert
=\left\vert N-\frac{\sin^2\frac{Nx}{2}}{N\sin^2\frac{x}{2}}\right\vert
\nonumber\\
&\overset{(\rm{a})}{\approx}N-\frac{\frac{N^2x^2}{4}-\frac{N^4x^4}{48}}{\frac{Nx^2}{4}-\frac{Nx^4}{48}}
=N-\frac{12N-N^3x^2}{12-x^2}
\nonumber\\
&\overset{(\rm{b})}{\approx}N-(12N-N^3x^2)\left(\frac{1}{12}+\frac{x^2}{144}\right)
\nonumber\\
&\overset{(\rm{c})}{\approx}\frac{N(N^2-1)}{12}x^2\approx \frac{N^3}{12}x^2.
\label{beam gain loss approximation}
\end{align}
Here, $\mathrm{(a)}$ follows from the fact that BF gain loss is always positive by definition and the Taylor series approximation $\sin^2x \approx x^2-\frac{1}{3}x^4$; $\mathrm{(b)}$ follows from the Taylor series approximation $\frac{1}{12-x^2}\approx\frac{1}{12}+\frac{x^2}{144}$. In $\mathrm{(c)}$, after expanding the expression, we omitted the fourth-order terms of $x$. At optimal values of $\tau=\tau^\star$ and $\phi = \phi^\star$, we have $x=\frac{\pi\alpha}{f_{\sf c}}\sin^2\theta$ after substitution. Finally, substituting this $x$ value into Eq. \eqref{beam gain loss approximation} yields Eq. \eqref{Beam gain loss}. Note that this approximation is highly accurate since the value of $x=\frac{\pi\alpha}{f_{\sf c}}\sin^2\theta$ is very small for the derived optimal values of $\tau^\star$ and $\phi^\star$. 
For instance, when $a_{\sf sat}=500$ km, $f_{\sf c}=f_{\sf p}=20$ GHz, $\beta=10^\circ$, and $N=64$, the true and approximated worst-case BF gain loss $\Delta\left(\pm \left(\frac{\pi}{2}-\beta\right)\right)$ are only $1.1672\times10^{-4}$ and $1.1675\times10^{-4}$, respectively, while the maximum BF gain is 64. This confirms that the BF gain loss of the designed rainbow beamformer is negligible and that the approximation of the presented BF gain loss formula is highly accurate.
\end{proof} \tcb{Note that near-full-gain BF across the entire FOV without beam sweeping is possible because rainbow BF achieves frequency-domain beam sweeping by varying beam directions with frequency. In contrast, conventional BF with frequency-independent beam directions requires time-domain beam sweeping for satellite acquisition within the FOV.}

\section{One-Shot Estimation Methods}
Our objective is to estimate the angular positions of multiple satellites from a single pilot signal reception, referred to as \emph{one-shot estimation}, using only a single time slot and without any prior knowledge. Let the measurement vector of the pilot be defined as $\mathbf{y}\triangleq [y[0], \cdots, y[L-1]]^{\sf T}$ and the AOA vector as $\boldsymbol{\Theta}\triangleq[\theta_1,\cdots,\theta_K]^{\sf T}$. The received signal can then be expressed as:
\begin{align}
\mathbf{y} = \mathbf{F}(\boldsymbol{\Theta}) \mathbf{s} + \mathbf{z},
\label{RX signal Matrix form}
\end{align}
where $\mathbf{F}(\boldsymbol{\Theta}) = [\mathbf{f}(\theta_1), \mathbf{f}(\theta_2), \dots, \mathbf{f}(\theta_K)]\in\mathbb{C}^{L\times K}$ with $\mathbf{f}(\theta_k) \triangleq [1, e^{j 2\pi \frac{f_{\sf D}(\theta_k)}{f_{\sf s}}}, \dots, e^{j 2 \pi \frac{f_{\sf D}(\theta_k)}{f_{\sf s}}(L-1) }]^{\sf T}\in\mathbb{C}^{L\times1}$ being the complex sinusoid sequence for the $k$-th satellite's Doppler shift; $\mathbf{s} = [\tilde{g}_1 x_1, \dots, \tilde{g}_K x_K]^{\sf T} \in \mathbb{C}^{K\times1}$ is received pilot symbol; ${\mathbf{z}}\triangleq[z[0],\cdots,z[L-1]]^{\sf T}\in \mathbb{C}^{L\times1}$ is the Gaussian noise vector. We here present three Doppler-aware angle estimation methods suitable for this one-shot scenario.

\subsection{FFT-based Estimation}
A simple yet powerful approach to estimate the Doppler frequencies from the received signal $\mathbf{y}$ is to use the fast Fourier transform (FFT). We extract \tcr{$K$ dominant peaks} from FFT of measurement $\mathbf{y}$ to estimate the Doppler frequencies $\{\hat{f}_{\sf D}(\theta_k)\}_{k=1}^K$.
Then, the angular position of the satellite can then be estimated using the relationship in Eq. \eqref{Doppler equation}:\footnote{\tcb{When Doppler frequency measurements alone are insufficient to uniquely determine satellite positions, supplementary information, such as satellite ephemeris data or the frequency-dependent beam pattern of the rainbow beamformer, can be employed to resolve the resulting ambiguity.}}
\begin{equation}
\hat{\theta}_k = -\arcsin\frac{\hat{f}_{\sf D}(\theta_k)}{\alpha}, \,\, k \in \{1, 2, \ldots, K\}.
\label{angle from doppler}
\end{equation}
\tcr{While this approach is computationally efficient, the FFT bin spacing of $\frac{f_s}{L}$ \cite{oppenheim1999discrete} may be too coarse for accurate satellite angle estimation, especially when $L$ is small or when multiple satellites are closely spaced. While zero-padding can reduce the FFT bin spacing and allow more precise spectral peak locations, it increases computational complexity because the FFT size grows with the padding length.} 

\subsection{Maximum Likelihood Estimation}
To achieve higher accuracy, we employ a maximum likelihood estimator (MLE). The angle estimates of satellites $\boldsymbol{\Theta}$ that maximizes the log-likelihood function is given by
\begin{equation}
\hat{\mathbf{\Theta}} = \arg\max_{\mathbf{\Theta}}\,\, \mathbf{y}^{\sf H} \mathbf{F}(\mathbf{\Theta})(\mathbf{F}(\mathbf{\Theta})^{\sf H}\mathbf{F}(\mathbf{\Theta}))^{-1}\mathbf{F}(\mathbf{\Theta})^{\sf H}\mathbf{y}.
\end{equation}
Here, the derivation process follows the typical MLE procedure and is omitted for brevity \cite{kay1993fundamentals}. In the case of a single satellite ($K=1$), the MLE is simplified as
\begin{equation}
\hat{\theta} = \arg\max_{\theta}\,\, \vert \mathbf{f}(\theta)^{\sf H}\mathbf{y}\vert,\,\,\theta\in\boldsymbol{\Psi}.
\end{equation}
Although MLE provides asymptotically optimal performance as $L\rightarrow\infty$ \cite{kay1993fundamentals}, it requires an exhaustive search over the parameter space $\boldsymbol{\Psi}$ since no closed-form solution exists. This makes MLE computationally inefficient, especially for cases with many satellites.

\subsection{Root-MUSIC Estimation}
\tcb{Doppler frequency can also be estimated through conventional multiple signal classification (MUSIC) algorithms, such as root-MUSIC \cite{rao1989performance}, which exploit the orthogonality between signal and noise subspaces.} However, there exists a challenge, as our goal is one-shot estimation using only a single time slot. If we compute the sample covariance as $\mathbf{R}_{{\mathbf{y}}}={\mathbf{y}}{\mathbf{y}}^{\sf H}$, $\mathbf{R}_{{\mathbf{y}}}$ is rank-1 and the eigenstructure-based estimation fails. To address this challenge, we compute the sample covariance using the \emph{spatial smoothing} technique \cite{Spatial_Smoothing_MUSIC} as
\begin{align}
\mathbf{R}_{\mathbf{y}} = \sum_{i=0}^{L-M}\mathbf{y}_{i:i+M-1} \mathbf{y}_{i:i+M-1}^H,
\end{align}
\tcb{where $\mathbf{y}_{i:i+M-1}=[y[i],\cdots,y[i+M-1]]^{\sf T}$ denotes the clipped measurement vector. This technique partitions the single measurement vector $\mathbf{y}$ into $(L-M+1)$ overlapping sub-vectors, creating multiple virtual snapshots from a single received signal. With an appropriately chosen value of $M$, this spatial smoothing approach effectively resolves the rank deficiency \cite{Spatial_Smoothing_MUSIC}.} Next, $\mathbf{R}_\mathbf{y}$ can be eigendecomposed as follows:
\begin{align}
\mathbf{R}_{\mathbf{y}} = \mathbf{E}_{\sf s} \mathbf{\Lambda}_{\sf s} \mathbf{E}_{\sf s}^{\sf H} + \mathbf{E}_{\sf n} \mathbf{\Lambda}_{\sf n} \mathbf{E}_{\sf n}^{\sf H},
\end{align}
where column vectors in ${\mathbf{E}_{\sf s}}$ and ${\mathbf{E}_{\sf n}}$ span the signal and noise subspace, respectively. $\mathbf{\Lambda}_{\sf s}$ and $\mathbf{\Lambda}_{\sf n}$ are the corresponding eigenvalue matrices. The noise subspace is nearly orthogonal to the complex sinusoid sequence vectors $\mathbf{f}(\theta_k)$, i.e., $\mathbf{f}(\theta_k)^{\sf H} \mathbf{E}_{\sf n} \approx \mathbf{0}$ for $k \in \{1,\cdots,K\}$. To find such $\{\theta_k\}_{k=1}^K$, we define the noise subspace polynomial:
\begin{align}
g(z) = \mathbf{a}(z)^{\sf H} \mathbf{E}_{\sf n} \mathbf{E}_{\sf n}^{\sf H} \mathbf{a}(z),
\label{root-MUSIC_equation}
\end{align}
where $\mathbf{a}(z) = [1, z, z^2, \dots, z^{M-1}]^{\sf T}$. We then numerically find the $K$ roots of $g(z)=0$ that are nearest to the unit circle. By denoting the $k$-th root as $z_k$, the Doppler frequency can be computed as $\hat{f}_{\sf D}(\theta_k)=\frac{f_{\sf s}}{2\pi}\angle z_k$. Finally, the AOA of the received signal can be computed using Eq. \eqref{angle from doppler}.
\begin{remark}{\rm \textbf{(Algorithm Complexity Comparison):}} \tcb{The big-O complexities of the presented estimation algorithms are listed in ascending order as follows:} $\mathcal{O}(L \log L)$ for FFT-based estimation involving FFT operations; $\mathcal{O}(L^3)$ for Root-MUSIC involving sample covariance computation, eigen-decomposition, and root-finding; $\mathcal{O}(N_{\sf grid}^K(L + K^3))$ for MLE where $N_{\sf grid}$ is the number of angular grid points and the complexity involves matrix inversion and grid search. Notably,  MLE suffers from poor scalability in multi-satellite scenarios as its complexity grows exponentially with the number of satellites $K$.\end{remark}

\tcb{The three algorithms represent exemplary methods that demonstrate the proposed framework's feasibility and evaluate its performance across different complexity-accuracy trade-offs. In practical implementations, these algorithm options enable adaptive selection based on specific operational requirements and available computational resources.}
\section{Simulation Results} 

\begin{figure}[t]
\centering
\includegraphics[width=0.7\linewidth]{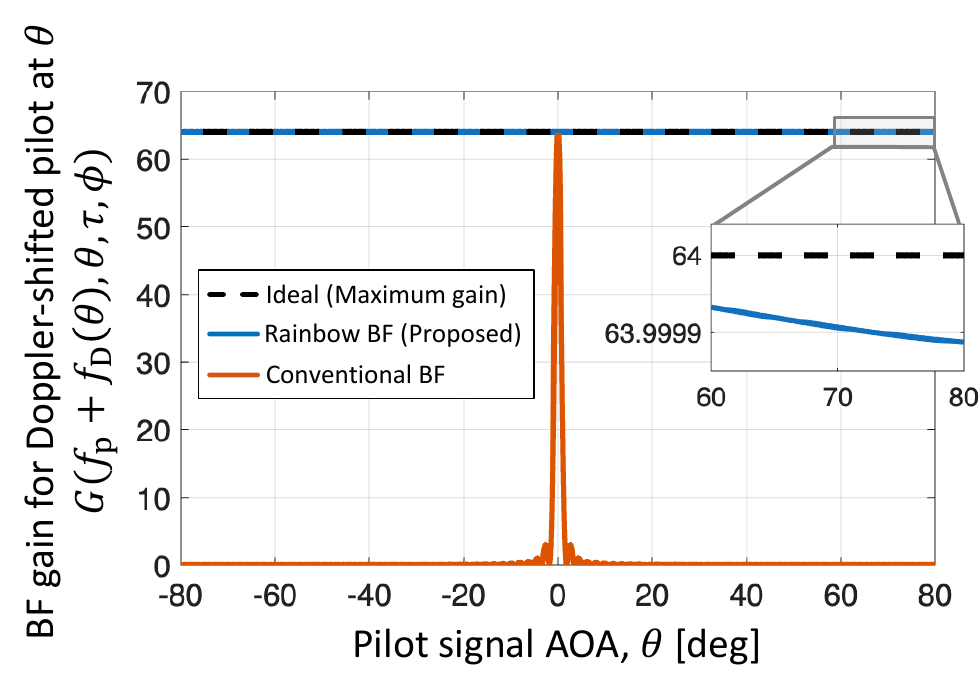}
\caption{\tcb{BF gain as a function of AOA for Doppler-shifted pilot signals.}}
\label{Simul_BF_gain}
\vspace{-3mm}
\end{figure}

\begin{figure}[t]
\centering
\includegraphics[width=0.7\linewidth]{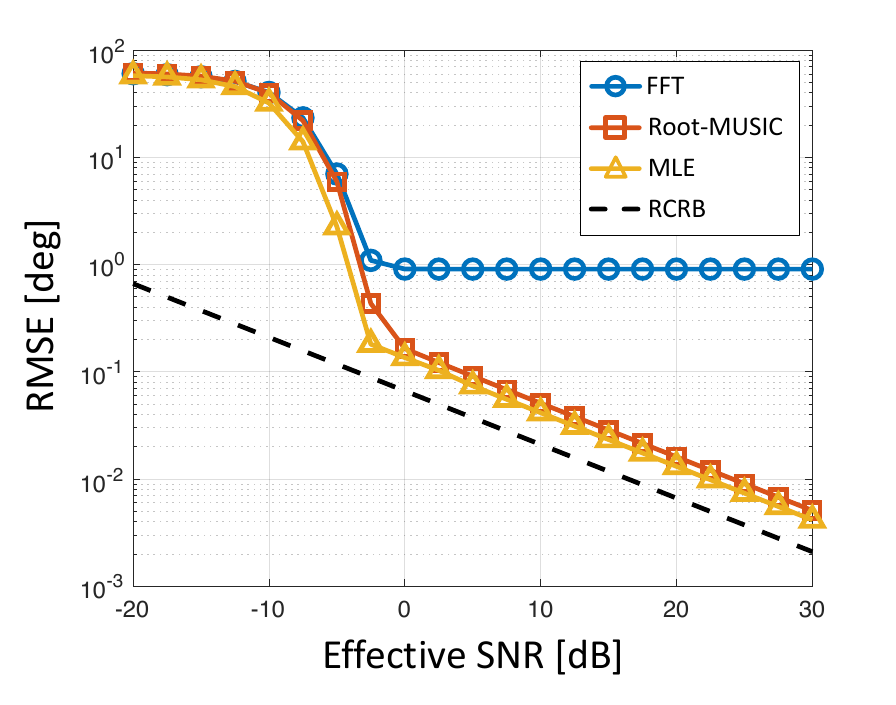}
\caption{RMSE performance of proposed rainbow BF-empowered satellite acquisition for a single satellite at \num{-45}$^\circ$ in a single time slot. }
\label{Simul_SingleSAT}
\vspace{-3mm}
\end{figure}

\begin{figure}[t]
\centering
\includegraphics[width=0.7\linewidth]{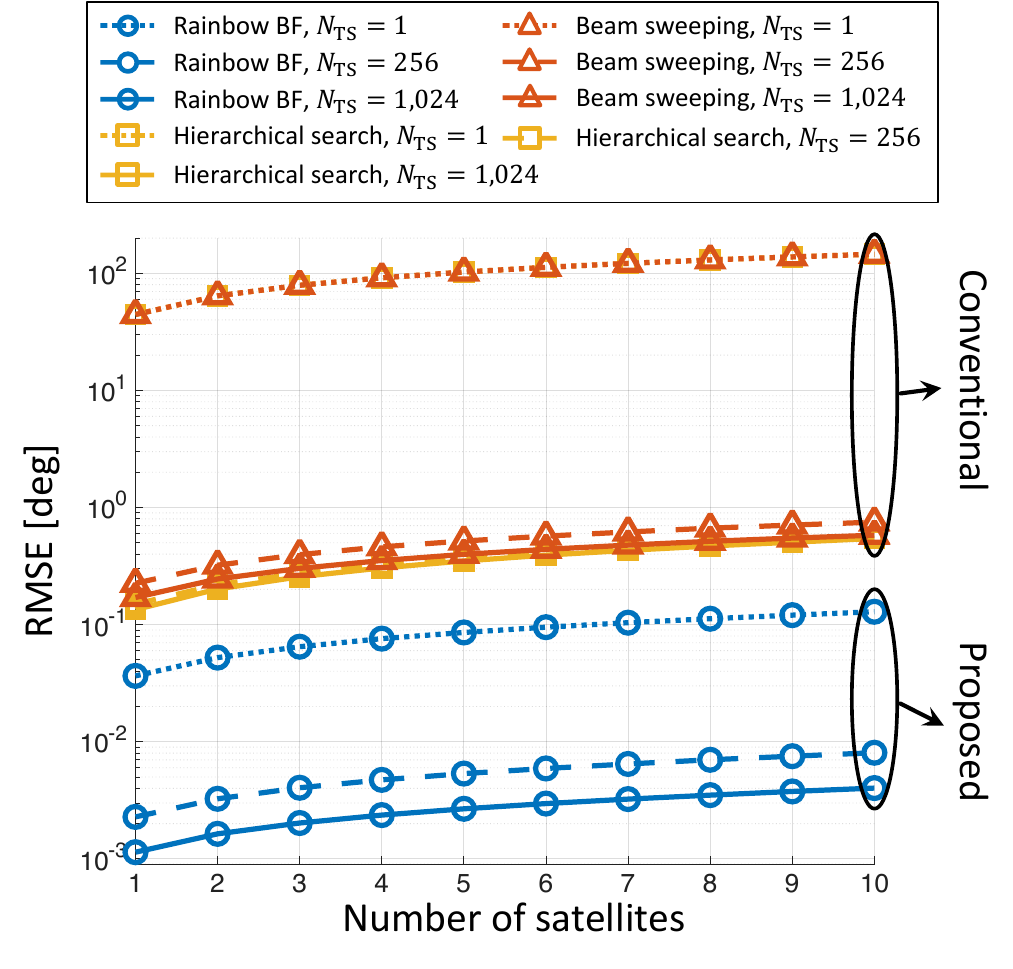}
\caption{\tcbb{RMSE performance comparison among conventional beam sweeping, hierarchical beam sweeping, and proposed rainbow BF-based satellite acquisition for multiple satellites with ${\sf SNR}_k=15$ dB.}}
\vspace{-3mm}
\label{Simul_MultiSAT}
\end{figure}

To assess our proposed approach, we consider the following simulation parameters: satellite altitude $a_{\sf sat}=500 {\rm ~km}$, center and pilot frequency $f_{\sf c}=f_{\sf p}=20 {\rm~GHz}$, \tcb{sampling frequency $f_{\sf s}=1 {\rm~MHz}$}, minimum elevation angle $\beta=10^\circ$, number of samples per time slot $L=64$, smoothing window size $M=32$, and the number of antenna elements $N=64$. The 3 dB beamwidth is approximately $1.6^\circ$. \tcb{The effective signal-to-noise ratio (SNR) of the signal from the $k$-th satellite is defined as 
\begin{align}
    {\sf SNR}_k = \frac{\mathbb{E}[|g_k(f_{\sf p}+f_{\sf D}(\theta_k))|^2]\mathbb{E}[|x_k|^2]NL}{\sigma^2}, \,\forall k.
\end{align}}

\vspace{-3mm}
\subsection{\tcb{Beamforming Performance}}
\tcb{Fig. \ref{Simul_BF_gain} shows the BF gain $G\left(f_{\sf p}+f_{\sf D}(\theta),\theta,\tau,\phi\right)$ for the signal received from a satellite located at angle $\theta$, where maintaining this value at $N$ for all $\theta$ was the beamformer design goal in Sec. III. Conventional BF can only amplify signals received from a very narrow angular range at a specific time instant. In contrast, the rainbow beamformer achieves near-maximum BF gain for pilot signals received with distinct Doppler shifts, regardless of the signals' AOA. Furthermore, even though we accept beam gain degradation during beamformer design, this degradation is negligible across all angular regions. Therefore, rainbow BF can amplify Doppler-shifted pilot signals from any direction without beamformer updates. In contrast, conventional BF requires beam sweeping over time.}

\vspace{-3mm}
\subsection{\tcb{Satellite Acquisition Performance}}
Fig. \ref{Simul_SingleSAT} compares the one-shot (single time slot) estimation performance of the three estimation algorithms with the root Cramér-Rao bound (RCRB) for a single satellite scenario. The RCRB represents the lower bound of the standard deviation for any unbiased estimator \cite{kay1993fundamentals}. \tcb{The algorithms achieve good root mean squared error (RMSE) in the order of MLE, root-MUSIC, and FFT, which is the reverse order of their computational complexity.} This demonstrates a clear trade-off between estimation algorithm complexity and estimation accuracy. \tcb{In practice, one of these algorithms can be selected based on available computational resources and operational requirements.}

\tcbb{Next, we compare the performance of three approaches: conventional beam sweeping, hierarchical beam search, and the proposed rainbow BF method, as shown in Fig. \ref{Simul_MultiSAT}. In conventional beam sweeping, the FOV is uniformly divided into $N_{\sf TS}$ angular segments with beam step size $\Delta\theta= \frac{(\pi-2\beta)}{N_{\sf TS}}$, where $N_{\sf TS}$ denotes the total number of available time slots. The beam direction at the $m$-th time slot is given by $\theta_m^{\sf sweep} = -\frac{\pi}{2}+\beta +(m-\frac{1}{2})\Delta\theta,$ and satellite directions are estimated based on peak received power. The hierarchical approach \cite{7390101} adopts a two-stage coarse-to-fine search strategy. In Stage 1, $\frac{N_{\sf TS}}{2}$  time slots are allocated for coarse sweeping over the entire FOV with step size $2\Delta\theta$, yielding $K$ candidate regions $\{\hat{\theta}_k^{\sf coarse}\}_{k=1}^K$ via peak detection. In Stage 2, the remaining $\frac{N_{\sf TS}}{2}$ time slots are distributed among the candidates, allocating $\lfloor \frac{N_{\sf TS}}{2K} \rfloor$ slots per candidate to perform fine sweeps within local regions $[\hat{\theta}_k^{\sf coarse} - \Delta\theta, \hat{\theta}_k^{\sf coarse} + \Delta\theta]$ for each $k$. Here, the 3 dB beamwidth is approximately $1.6^\circ$. For the proposed rainbow BF, time-averaging of the measurement $\mathbf{y}$ over $N_{\sf TS}$ time slots is performed for fair comparison. When $N_{\sf TS}=1$, both beam sweeping and hierarchical search yield a fixed estimate at $0^\circ$. We compute the error in each realization as the sum of estimation errors across all satellites. }

\tcbb{The results confirm that the proposed rainbow BF-empowered satellite acquisition accurately estimates multiple satellites using only a single time slot. Moreover, rainbow BF significantly outperforms both conventional and hierarchical beam sweeping approaches. This substantial performance gap originates from a fundamental difference in search mechanisms. Unlike time-domain searching methods that rely on sequential received-power scanning, rainbow BF enables a fully parallel search across the entire angular space via frequency-domain beam sweeping. In particular, different frequency components simultaneously probe distinct angular regions by exploiting angle-dependent Doppler shifts and beam-squint effects. Notably, rainbow BF with a single time slot achieves performance comparable to or better than hierarchical sweeping with $N_{\sf TS}=1{,}024$, thereby validating the feasibility of accurate one-shot multi-satellite acquisition.}

\section{Conclusion}\label{Sec6}
This correspondence introduced a novel one-shot satellite acquisition framework that leverages Doppler shifts and beam-squint effects rather than compensating for them. A closed-form rainbow beamformer was developed to enable simultaneous multi-satellite reception with near-maximum BF gain. We also presented three viable satellite angle estimation algorithms. Simulation results confirm that our approach significantly outperforms conventional beam sweeping methods, achieving accurate and efficient one-shot multi-satellite acquisition. \tcb{Future work will extend the proposed approach to three-dimensional geometry to enhance practical applicability.}

\begingroup
\bibliographystyle{IEEEtran}
\bibliography{bib_arXiv}

\begin{thebibliography}{10}
\providecommand{\url}[1]{#1}
\csname url@samestyle\endcsname
\providecommand{\newblock}{\relax}
\providecommand{\bibinfo}[2]{#2}
\providecommand{\BIBentrySTDinterwordspacing}{\spaceskip=0pt\relax}
\providecommand{\BIBentryALTinterwordstretchfactor}{4}
\providecommand{\BIBentryALTinterwordspacing}{\spaceskip=\fontdimen2\font plus
\BIBentryALTinterwordstretchfactor\fontdimen3\font minus \fontdimen4\font\relax}
\providecommand{\BIBforeignlanguage}[2]{{%
\expandafter\ifx\csname l@#1\endcsname\relax
\typeout{** WARNING: IEEEtran.bst: No hyphenation pattern has been}%
\typeout{** loaded for the language `#1'. Using the pattern for}%
\typeout{** the default language instead.}%
\else
\language=\csname l@#1\endcsname
\fi
#2}}
\providecommand{\BIBdecl}{\relax}
\BIBdecl

\bibitem{kodheli2020satellite}
O.~Kodheli \emph{et~al.}, ``Satellite communications in the new space era: A survey and future challenges,'' \emph{IEEE Commun. Surv. Tutor.}, vol.~23, no.~1, pp. 70--109, 2021.

\bibitem{yang2023starlink}
C.~Yang and A.~Soloviev, ``Starlink {D}oppler and {D}oppler rate estimation via coherent combining of multiple tones for opportunistic positioning,'' in \emph{Proc. IEEE/ION Posit. Locat. Navig. Symp. (PLANS)}, 2023, pp. 1143--1153.

\bibitem{1638663}
T.~Wang, J.~G. Proakis, E.~Masry, and J.~R. Zeidler, ``Performance degradation of {OFDM} systems due to {D}oppler spreading,'' \emph{IEEE Trans. Wirel. Commun.}, vol.~5, no.~6, pp. 1422--1432, 2006.

\bibitem{10464931}
A.~K. Meshram, S.~Kumar, J.~Querol, and S.~Chatzinotas, ``Doppler effect mitigation in {LEO}-based {5G} non-terrestrial networks,'' in \emph{Proc. IEEE Global Commun. Conf. (GLOBECOM) Workshops}, 2023, pp. 311--316.

\bibitem{dai2022delay}
L.~Dai, J.~Tan, Z.~Chen, and H.~V. Poor, ``Delay-phase precoding for wideband {THz} massive {MIMO},'' \emph{IEEE Trans. Wirel. Commun.}, vol.~21, no.~9, pp. 7271--7286, 2022.

\bibitem{you2022beam}
L.~You \emph{et~al.}, ``Beam squint-aware integrated sensing and communications for hybrid massive {MIMO} {LEO} satellite systems,'' \emph{IEEE J. Sel. Areas Commun.}, vol.~40, no.~10, pp. 2994--3009, 2022.

\bibitem{mewes2025beam}
T.~Mewes, W.~Rave, and G.~Fettweis, ``Beam squint reduction by combining phased subarrays with true time delays in baseband,'' \emph{IEEE Wirel. Commun. Lett.}, vol.~14, no.~4, pp. 1159--1163, 2025.

\bibitem{ratnam2022joint}
V.~V. Ratnam, J.~Mo, A.~Alammouri, B.~L. Ng, J.~Zhang, and A.~F. Molisch, ``Joint phase-time arrays: A paradigm for frequency-dependent analog beamforming in {6G},'' \emph{IEEE Access}, vol.~10, pp. {733}{64}--{733}{77}, 2022.

\bibitem{10914514}
A.~Wadaskar, V.~Boljanovic, H.~Yan, and D.~Cabric, ``Fast {3D} beam training with true-time-delay arrays in wideband millimeter-wave systems,'' \emph{IEEE Trans. Wirel. Commun.}, vol.~24, no.~6, pp. 5146--5162, 2025.

\bibitem{li2024stacked}
Q.~Li, M.~El-Hajjar, C.~Xu, J.~An, C.~Yuen, and L.~Hanzo, ``Stacked intelligent metasurfaces for holographic {MIMO}-aided cell-free networks,'' \emph{IEEE Trans. Commun.}, vol.~72, no.~11, pp. 7139--7151, 2024.

\bibitem{li2025holographic}
Q.~Li, M.~El-Hajjar, K.~Cao, C.~Xu, H.~Haas, and L.~Hanzo, ``Holographic metasurface-based beamforming for multi-altitude {LEO} satellite networks,'' \emph{IEEE Trans. Wirel. Commun.}, vol.~24, no.~4, pp. 3103--3116, 2025.

\bibitem{you2020massive}
L.~You, K.-X. Li, J.~Wang, X.~Gao, X.-G. Xia, and B.~Ottersten, ``Massive {MIMO} transmission for {LEO} satellite communications,'' \emph{IEEE J. Sel. Areas Commun.}, vol.~38, no.~8, pp. 1851--1865, 2020.

\bibitem{10680149}
J.~Seong, J.~Park, D.-H. Jung, J.~Park, and W.~Shin, ``Rate-splitting for joint unicast and multicast transmission in {LEO} satellite networks with non-uniform traffic demand,'' \emph{IEEE J. Sel. Areas Commun.}, vol.~43, no.~1, pp. 122--138, 2025.

\bibitem{10694172}
A.~B. De~F.~Diniz, T.~Eriksson, and U.~Gustavsson, ``Doppler shift estimation for satellite communications using linear estimators,'' in \emph{Proc. IEEE Int. Workshop Signal Process. Adv. Wireless Commun. (SPAWC)}, 2024, pp. 686--690.

\bibitem{oppenheim1999discrete}
A.~V. Oppenheim, \emph{Discrete-time signal processing}.\hskip 1em plus 0.5em minus 0.4em\relax Pearson Education India, 1999.

\bibitem{kay1993fundamentals}
S.~M. Kay, \emph{Fundamentals of statistical signal processing: Estimation theory}.\hskip 1em plus 0.5em minus 0.4em\relax Prentice-Hall, Inc., 1993.

\bibitem{rao1989performance}
B.~D. Rao and K.~S. Hari, ``Performance analysis of root-{MUSIC},'' \emph{IEEE Trans. Acoust., Speech, Signal Process.}, vol.~37, no.~12, pp. 1939--1949, 1989.

\bibitem{Spatial_Smoothing_MUSIC}
Q.~Chen and R.~Liu, ``On the explanation of spatial smoothing in {MUSIC} algorithm for coherent sources,'' in \emph{Proc. IEEE Int. Conf. Inf. Sci. Technol. (ICIST)}, 2011, pp. 699--702.

\bibitem{7390101}
Z.~Xiao, T.~He, P.~Xia, and X.-G. Xia, ``Hierarchical codebook design for beamforming training in millimeter-wave communication,'' \emph{IEEE Trans. Wirel. Commun.}, vol.~15, no.~5, pp. 3380--3392, 2016.

\end{thebibliography}
\endgroup

\end{document}